\documentclass[runningheads]{llncs}

\newtheorem{mytheorem}{Theorem}

\newtheorem{myexample}[mytheorem]{Example}

\newtheorem{mydefinition}[mytheorem]{Definition}
\newtheorem{myprinciple}[mytheorem]{Principle}

\begin{document}

\title{AGI and the Knight-Darwin Law: why idealized AGI reproduction requires
collaboration}
\titlerunning{AGI and the Knight-Darwin Law}

\author{Samuel Allen Alexander\inst{1}\orcidID{0000-0002-7930-110X}}
\authorrunning{S.\ A.\ Alexander}
\institute{The U.S.\ Securities and Exchange Commission
\email{samuelallenalexander@gmail.com}
\url{https://philpeople.org/profiles/samuel-alexander/publications}}

\maketitle

\begin{abstract}
Can an AGI create a more intelligent AGI?
Under idealized assumptions, for a certain theoretical
type of intelligence, our answer is:
``Not without outside help''.
This is a paper on the mathematical structure of AGI populations
when parent AGIs create child AGIs. We argue that such populations
satisfy a certain biological law.
Motivated by observations of sexual reproduction in seemingly-asexual
species,
the Knight-Darwin Law states that it is impossible for one organism
to asexually produce another, which asexually produces another,
and so on forever:
that any sequence of organisms (each one a child of the previous) must contain
occasional multi-parent organisms, or must terminate.
By proving that a certain measure (arguably an intelligence measure)
decreases when an
idealized parent AGI single-handedly creates a child AGI,
we argue that a similar Law holds for AGIs.

\end{abstract}

\section{Introduction}

It is difficult to reason about agents with Artificial General Intelligence (AGIs)
programming AGIs\footnote{Our approach to AGI is what
Goertzel \cite{goertzel2014artificial} describes
as the Universalist Approach:
we consider ``...an idealized case of AGI, similar to
assumptions like the frictionless plane in physics'', with the hope that by
understanding this ``simplified special
case, we can use the understanding we've gained to address more realistic
cases.''}. To get our hands on something solid,
we have attempted to find structures that abstractly capture the core essence of
AGIs programming AGIs. This led us to discover
what we call the \emph{Intuitive Ordinal Notation System} (presented in Section
\ref{notationsystemsection}), an ordinal notation system that gets directly at
the heart of AGIs creating AGIs.

We call an AGI \emph{truthful} if the things it knows
are true\footnote{Knowledge and truth are formally
treated in \cite{alexander2019measuring} but here we aim at a more general audience.
For the purposes of this paper, an AGI can be thought of as knowing a fact if and only
if the AGI would list that fact if commanded to spend eternity listing all the facts
that it knows. We assume such knowledge is closed under deduction, an assumption which
is ubiquitous in modal logic, where it often appears in a form like
$K(\phi\rightarrow\psi)\rightarrow (K(\phi)\rightarrow K(\psi))$.
Of course,
it is only in the idealized context of this paper that one should assume AGIs
satisfy such closure.}.
In \cite{alexander2019measuring}, we argued
that if a truthful AGI $X$ creates (without external help) a truthful AGI $Y$ in such a way
that $X$ knows the truthfulness
of $Y$, then $X$ must be more intelligent than $Y$
in a certain formal sense. The argument is based on the key assumption that if $X$
creates $Y$, without external help, then $X$ necessarily knows $Y$'s source code.

Iterating the above argument, suppose $X_1,X_2,\ldots$
are truthful AGIs such that each $X_i$ creates, and knows the truthfulness and
the code of, $X_{i+1}$. Assuming the previous paragraph, $X_1$ would be more
intelligent than $X_2$, which would be more
intelligent than $X_3$, and so on (in our certain formal sense).
In Section \ref{informalargumentsection} we will argue that this implies
it is impossible for such
a list $X_1,X_2,\ldots$ to go on forever: it would have to stop after finitely
many elements\footnote{This may initially seem to contradict some mathematical
constructions \cite{kripke2019ungroundedness}
\cite{visser2002semantics} of infinite descending chains of theories. But those
constructions only work for weaker languages, making them inapplicable to
AGIs which comprehend
linguistically strong second-order predicates.}.

At first glance, the above results might
seem to suggest skepticism regarding the singularity---regarding
what Hutter \cite{hutter2012} calls \emph{intelligence explosion}, the idea of
AGIs creating better AGIs, which create even better AGIs, and so on.
But there is a loophole (discussed further in
Section \ref{knightdarwinagisection}). Suppose AGIs $X$ and $X'$
collaborate to create $Y$. Suppose $X$
does part of the programming work,
but keeps the code secret from $X'$, and suppose $X'$ does another part of
the programming work, but keeps the code secret from $X$.
Then neither
$X$ nor $X'$ knows $Y$'s full source code, and yet if $X$ and $X'$ trust
each other, then both $X$ and $X'$ should be able to trust $Y$, so the above-mentioned
argument breaks down.

Darwin and his contemporaries observed that even
seemingly asexual plant species occasionally reproduce sexually.
For example, a plant in which pollen is ordinarily isolated, might
release pollen into the air if a storm damages the part of the
plant that would otherwise shield the pollen\footnote{Even prokaryotes can
be considered to occasionally have multiple parents, if lateral gene
transfer is taken into account.}.
The Knight-Darwin Law \cite{darwin1898knight}, named after Charles Darwin
and Andrew Knight, is the
principle (rephrased in modern language) that there cannot be an infinite
sequence $X_1,X_2,\ldots$ of biological organisms such that each $X_i$ asexually
parents $X_{i+1}$. In other words, if $X_1,X_2,\ldots$ is any infinite list of
organisms such that each $X_i$ is a biological parent of $X_{i+1}$, then some of the
$X_i$ would need to be multi-parent organisms.
The reader will immediately notice a striking parallel between
this principle and the discussion in the previous two paragraphs.

In Section \ref{notationsystemsection} we present the Intuitive Ordinal Notation
System.

In Section \ref{informalargumentsection} we argue\footnote{This argument appeared
in a fully rigorous form in \cite{alexander2019measuring},
but in this paper we attempt to make it more approachable.} that if truthful AGI $X$ creates
truthful AGI $Y$, such that $X$
knows the code and truthfulness of $Y$, then, in a certain formal sense, $Y$ is
less intelligent
than $X$.

In Section \ref{knightdarwinagisection} we adapt the Knight-Darwin Law from biology to AGI
and speculate about what it might mean for AGI.

In Section \ref{objectionsection} we address some anticipated objections.


Sections \ref{notationsystemsection}--\ref{informalargumentsection} are not new
(except for new motivation and discussion). Their content appeared in
\cite{alexander2019measuring}, and was more rigorously formalized there.
Sections \ref{knightdarwinagisection}--\ref{objectionsection} contain this paper's
new material. Of this, some was hinted at in \cite{alexander2019measuring},
and some appeared (weaker and less approachably) in the author's
dissertation \cite{alexanderdissert}.

\section{The Intuitive Ordinal Notation System}
\label{notationsystemsection}

If humans can write AGIs, and AGIs are at least as smart as humans,
then AGIs should be capable of writing AGIs.
Based on the conviction that an AGI should be capable of writing AGIs,
we would like to come up with a more concrete structure, easier to reason
about, which we can use to better understand AGIs.

To capture the essence of an AGI's AGI-programming
capability, one might try: ``computer
program that prints computer programs.'' But this only
captures the
AGI's capability to write \emph{computer programs}, not to write \emph{AGIs}.

How
about: ``computer program that prints computer programs that print
computer programs''? This second attempt
seems to capture an AGI's ability to write \emph{program-writing programs},
not to write \emph{AGIs}.

Likewise, ``computer program that prints computer programs that print computer
programs that print computer programs'' captures the ability to write
\emph{program-writing-program-writing programs}, not \emph{AGIs}.

We need to short-circuit the above process. We need to come up with a notion
X which is equivalent to ``computer program that prints members of X''.

\begin{mydefinition}
\label{literalnotationdef}
    (See the following examples)
    We define the Intuitive Ordinal Notations to be the smallest set $\mathcal P$
    of computer programs such that:
    \begin{itemize}
        \item
            Each computer program $p$ is in $\mathcal P$ iff all of
            $p$'s outputs are also in $\mathcal P$.
    \end{itemize}
\end{mydefinition}

\begin{myexample}
\label{simpleexamples}
(Some simple examples)
    \begin{enumerate}
    \item
    Let $P_0$ be ``End'', a program which immediately stops without any outputs.
    Vacuously, all of $P_0$'s outputs are in $\mathcal P$
    (there are no such outputs). So $P_0$ is an Intuitive Ordinal Notation.
    \item
    Let $P_1$ be ``Print(`End')'', a program which outputs ``End'' and then
    stops. By (1), all of $P_1$'s outputs are Intuitive Ordinal Notations,
    therefore, so is $P_1$.
    \item
    Let $P_2$ be ``Print(`Print(`End')')'', which outputs ``Print(`End')'' and then
    stops. By (2), all of $P_2$'s outputs are Intuitive Ordinal Notations,
    therefore, so is $P_2$.
    \end{enumerate}
\end{myexample}

\begin{myexample}
\label{omegaexample}
(A more interesting example)
    Let $P_\omega$ be the program:
    \[
        \mbox{\normalfont Let X = `End';
        While(True) \{ Print(X); X = ``Print(`'' + X + ``')''; \}}
    \]
    When executed, $P_\omega$ outputs ``End'', ``Print(`End')'',
    ``Print(`Print(`End')')'', and so on forever. As
    in Example \ref{simpleexamples}, all of these are Intuitive Ordinal Notations.
    Therefore, $P_\omega$ is an Intuitive Ordinal Notation.
\end{myexample}

To make Definition \ref{literalnotationdef} fully rigorous, one would need
to work in a formal model of computation; see \cite{alexander2019measuring} (Section 3)
where we do exactly that.
Examples \ref{simpleexamples} and \ref{omegaexample} are reminiscent
of Franz's approach of ``head[ing] for general algorithms at low complexity levels
and fill[ing] the task cup from the bottom up'' \cite{franz2015toward}.
For a much larger collection of examples, see \cite{github}.
A different type of example will be sketched in the proof of Theorem
\ref{maintheorem} below.

\begin{mydefinition}
    For any Intuitive Ordinal Notation $x$, we define an ordinal $|x|$
    inductively as follows: $|x|$ is the smallest ordinal $\alpha$
    such that $\alpha>|y|$ for every output $y$ of $x$.
\end{mydefinition}

\begin{myexample}
    \begin{itemize}
        \item
        Since $P_0$ (from Example \ref{simpleexamples}) has no outputs,
        it follows that $|P_0|=0$, the smallest ordinal.
        \item
        Likewise, $|P_1|=1$ and $|P_2|=2$.
        \item
        Likewise, $P_\omega$ (from Example \ref{omegaexample}) has outputs
        notating $0, 1, 2, \ldots$---all the finite natural numbers. It follows
        that $|P_\omega|=\omega$, the smallest
        infinite ordinal.
        \item
        Let $P_{\omega+1}$ be the program ``Print($P_\omega$)'',
        where $P_\omega$ is as in Example \ref{omegaexample}.
        It follows that $|P_{\omega+1}|=\omega+1$, the next ordinal after
        $\omega$.
    \end{itemize}
\end{myexample}

The Intuitive Ordinal Notation System is a more intuitive simplification of
an ordinal notation system known as Kleene's $\mathcal O$.

\section{Intuitive Ordinal Intelligence}
\label{informalargumentsection}

Whatever an AGI is, an AGI should know certain
mathematical facts.
The following is a universal notion of an AGI's intelligence based
solely on said facts. In \cite{alexander2019measuring}
we argue that this notion captures key components of intelligence such as
pattern recognition, creativity, and the ability or
generalize. We will give further justification in
Section \ref{objectionsection}. Even if the reader refuses to accept this as
a genuine intelligence measure, that is merely a name we have chosen for it:
we could give it any other name without compromising this paper's structural
results.

\begin{mydefinition}
\label{maindefinition}
    The \emph{Intuitive Ordinal Intelligence} of a truthful AGI $X$ is the smallest
    ordinal $|X|$ such that $|X|>|p|$ for every Intuitive Ordinal Notation
    $p$ such that $X$ knows that $p$ is an
    Intuitive Ordinal Notation.
\end{mydefinition}

The following theorem provides a relationship\footnote{Possibly formalizing a
relationship implied
offhandedly by Chaitin, who suggests ordinal computation as a mathematical challenge
intended to encourage evolution, ``and the larger the ordinal,
the fitter the organism'' \cite{chaitin}.} between Intuitive Ordinal Intelligence
and AGI creation of AGI. Here, we give an informal version of the proof; for a version
spelled out in complete formal detail, see \cite{alexander2019measuring}.

\begin{mytheorem}
\label{maintheorem}
    Suppose $X$ is a truthful AGI, and $X$ creates a truthful AGI $Y$
    in such a way that $X$ knows $Y$'s code and truthfulness. Then
    $|X|>|Y|$.
\end{mytheorem}

\begin{proof}
    Suppose $Y$ were commanded to
    spend eternity enumerating the biggest Intuitive Ordinal Notations $Y$ could
    think of. This would result in some list $L$ of Intuitive Ordinal Notations
    enumerated by $Y$. Since $Y$ is an AGI, $L$ must be computable. Thus, there
    is some computer program
    $P$ whose outputs are exactly $L$.
    Since $X$ knows $Y$'s code,
    and as an AGI, $X$ is capable of reasoning about code,
    it follows that $X$ can infer a program $P$ that\footnote{For example,
    $X$ could write a general program $Sim(c)$ that simulates an input AGI $c$
    waking up in an empty room and being commanded
    to spend eternity enumerating Intuitive Ordinal Notations. This program $Sim(c)$
    would then output whatever outputs AGI $c$ outputs under those
    circumstances. Having written $Sim(c)$, $X$ could then obtain $P$ by
    pasting $Y$'s code into $Sim$ (a string operation---not actually running $Sim$
    on $Y$'s code).
    Nowhere in this process do we require $X$ to actually
    execute $Sim$ (which might be computationally infeasible).} lists $L$.
    Having constructed $P$ this way, $X$ knows: ``$P$ outputs
    $L$, the list of things $Y$ would output if $Y$ were commanded to spend eternity
    trying to enumerate large Intuitive Ordinal Notations''.
    Since $X$ knows $Y$ is truthful,
    $X$ knows that $L$ contains nothing except Intuitive Ordinal Notations,
    thus $X$ knows that $P$'s outputs are Intuitive Ordinal Notations,
    and so $X$ knows that $P$ is an Intuitive Ordinal Notation.
    So $|X|>|P|$. But $|P|$ is
    the least ordinal $>|Q|$ for all $Q$ output by $L$, in other words,
    $|P|=|Y|$.
    \qed
\end{proof}

Theorem \ref{maintheorem} is mainly intended for the situation where parent $X$ creates
independent child $Y$, but can also be applied in case $X$ self-modifies,
viewing the original $X$ as being replaced by the new self-modified
$Y$ (assuming $X$ has prior
knowledge of the code and truthfulness of the modified result).

It would be straightforward to extend Theorem \ref{maintheorem} to cases where $X$ creates
$Y$ non-deterministically. Suppose $X$ creates $Y$ using random numbers, such that $X$ knows
$Y$ is
one of $Y_1,Y_2,\ldots,Y_k$ but $X$ does not know which. If $X$ knows that $Y$ is truthful,
then $X$ must know that each $Y_i$ is truthful (otherwise, if some $Y_i$ were not truthful,
$X$ could not rule out that $Y$ was that non-truthful $Y_i$). So by Theorem \ref{maintheorem},
each $|Y_i|$ would be $<|X|$. Since $Y$ is one of the $Y_i$, we would still have
$|Y|<|X|$.

\section{The Knight-Darwin Law}
\label{knightdarwinagisection}

\begin{quote}
``...it is a general law of nature that no organic being self-fertilises itself
for a perpetuity of generations; but that a cross with another individual
is occasionally---perhaps at long intervals of time---indispensable.''
(Charles Darwin)
\end{quote}

In his Origin of Species, Darwin devotes many
pages to the above-quoted principle, later called the
Knight-Darwin Law \cite{darwin1898knight}. In \cite{alexander2013}
we translate
the Knight-Darwin Law into mathematical language.

\begin{myprinciple}
(The Knight-Darwin Law)
There cannot be an infinite sequence
$x_1,x_2,\ldots$ of organisms such that each $x_i$
is the lone biological parent of $x_{i+1}$.
If each $x_i$ is a parent of $x_{i+1}$, then some $x_{i+1}$
must have multiple parents.
\end{myprinciple}

A key fact about the ordinals is they
are \emph{well-founded}:
there is
no infinite sequence $o_1,o_2,\ldots$ of ordinals such that\footnote{This
is essentially true by definition,
unfortunately the formal definition of ordinal numbers is outside the scope of
this paper.} each
$o_i>o_{i+1}$. In Theorem \ref{maintheorem} we showed that if truthful
AGI $X$ creates truthful AGI $Y$ in such a way as to know the truthfulness
and code of $Y$, then $X$ has a higher Intuitive Ordinal Intelligence
than $Y$. Combining this with the well-foundedness of the ordinals yields
a theorem extremely similar to the Knight-Darwin Law.

\begin{mytheorem}
\label{maintheorem2}
(The Knight-Darwin Law for AGIs)
There cannot be an infinite sequence
$X_1,X_2,\ldots$ of truthful AGIs such that each $X_i$
creates $X_{i+1}$ in such a way as to know $X_{i+1}$'s truthfulness and code.
If each $X_i$ creates $X_{i+1}$
so as to know $X_{i+1}$ is truthful,
then
occasionally certain $X_{i+1}$'s must be
co-created by multiple creators (assuming that creation by
a lone creator implies the lone creator would know $X_{i+1}$'s code).
\end{mytheorem}

\begin{proof}
By Theorem \ref{maintheorem}, the Intuitive Ordinal Intelligence of $X_1,X_2,\ldots$
would be an infinite strictly-descending sequence of ordinals, violating
the well-foundedness of the ordinals.
\qed
\end{proof}

It is perfectly consistent with
Theorem \ref{maintheorem} that $Y$ might operate faster
than $X$, performing better in realtime environments (as in \cite{gavane}).
It may even be that $Y$ performs so much faster that it would be infeasible for
$X$ to use the knowledge of $Y$'s code to simulate $Y$.
Theorems \ref{maintheorem} and \ref{maintheorem2} are profound because
they suggest that descendants might initially appear more practical
(faster, better at problem-solving, etc.),
yet, without outside help, their knowledge must degenerate.
This parallels the \emph{hydra game} of
Kirby and Paris \cite{kirby1982accessible}, where a hydra
seems to grow as the player cuts off its heads, yet
inevitably dies if the player keeps cutting.

If AGI $Y$ has distinct parents $X$ and $X'$, neither of which fully knows
$Y$'s code, then Theorem \ref{maintheorem} does not apply to $X,Y$ or $X',Y$ and
does not force
$|Y|<|X|$ or $|Y|<|X'|$. This does not necessarily mean that
$|Y|$ can be arbitrarily large, though. If $X$ and $X'$ were themselves
created single-handedly by a lone parent $X_0$, similar
reasoning to Theorem \ref{maintheorem} would force $|Y|<|X_0|$ (assuming $X_0$
could infer the code and truthfulness of $Y$ from those of $X$ and $X'$)\footnote{This
suggests possible generalizations of the Knight-Darwin Law such as ``There cannot be
an infinite sequence $x_1,x_2,\ldots$ of biological organisms such that each $x_i$
is the lone grandparent of $x_{i+1}$,'' and AGI versions of same. This also raises
questions about the relationship between the set of AGIs initially created by
humans and how intelligent the offspring of those initial AGIs can be. These questions
go beyond the scope of this paper but perhaps they could be a fruitful area for future
research.}.

In the remainder of this section, we will non-rigorously speculate about three implications
Theorem \ref{maintheorem2} might have for AGIs and for AGI research.

\subsection{Motivation for Multi-agent Approaches to AGI}

If AGI ought to be capable of programming AGI,
Theorem \ref{maintheorem2} suggests
that a fundamental aspect of AGI should be the ability to collaborate with other
AGIs in the creation of new AGIs.
This seems to suggest there should be no
such thing as a \emph{solipsistic} AGI\footnote{That is, an AGI which believes itself
to be the only entity in the universe.}, or at least, solipsistic AGIs would be
limited in their reproduction ability.
For, if an AGI were solipsistic, it
seems like it would be difficult for this AGI to collaborate with other AGIs
to create child AGIs.
To quote Hern{\'a}ndez-Orallo et al: ``The appearance of multi-agent systems is a sign that
the future of machine intelligence will not be found in monolithic systems
solving tasks without other agents to compete or collaborate with''
\cite{hernandez2011more}.

More practically,
Theorem \ref{maintheorem2} might suggest
prioritizing research on multi-agent approaches to AGI, such as
\cite{castelfranchi1998modelling}, \cite{hernandez2011more},
\cite{hibbard2011societies},
\cite{kolonin2018reputation},
\cite{potyka2016group},
\cite{thorisson2004constructionist},
and similar work.

\subsection{Motivation for AGI Variety}

Darwin used the Knight-Darwin Law as a foundation for
a broader thesis that the survival of a
species depends on the inter-breeding of many members.
By analogy, if our goal is to create robust AGIs, perhaps
we should focus on creating a wide variety of AGIs, so that
those AGIs can co-create more AGIs.

On the other hand, if we want to reduce the danger of AGI getting out of control,
perhaps we should \emph{limit} AGI variety. At the extreme end
of the spectrum, if humankind were to limit itself to only creating one single
AGI\footnote{Or to perfectly isolate
different AGIs away from
one another---see \cite{yampolskiy2012leakproofing}.}, then
Theorem \ref{maintheorem2} would constrain the extent to which
that AGI could reproduce.

\subsection{AGI Genetics}

If AGI collaboration
is a fundamental requirement for AGI ``populations'' to propagate, it might
someday be possible to view AGI through a genetic lens. For example, if AGIs $X$ and $X'$
co-create child $Y$,
if $X$ runs operating
system $O$, and $X'$ runs operating system $O'$, perhaps $Y$
will somehow exhibit traces of both $O$ and $O'$.

\section{Discussion}
\label{objectionsection}

In this section, we discuss some anticipated objections.

\subsection{What does Definition \ref{maindefinition} really have to do with intelligence?}

We do not claim that Definition \ref{maindefinition} is the ``one true measure'' of
intelligence. Maybe there is no such thing: maybe intelligence is inherently
multi-dimensional. Definition \ref{maindefinition} measures a type of
intelligence based on mathematical knowledge\footnote{Wang has
correctly pointed out \cite{wang2007} that an AGI consists of much more than merely
a knowledge-set of mathematical facts. Still, we feel mathematical knowledge is at least
one important aspect of an AGI's intelligence.} closed under logical deduction. An
AGI could be good at problem-solving
but poor at ordinals. But the broad AGIs we are talking about in this paper
should be capable (if properly
instructed) of attempting any reasonable well-defined task, including that of
notating ordinals. So Definition \ref{maindefinition} does
measure one aspect of an AGI's abilities. Perhaps
a word like
``mathematical-knowledge-level'' would fit better: but
that would not change
the Knight-Darwin Law implications.

Intelligence has core components like pattern-matching,
creativity, and the ability to generalize.
We claim that these components are needed if one wants to
competitively name large ordinals. If $p$ is an Intuitive Ordinal Notation
obtained using certain facts and techniques, then \emph{any} AGI who used those
facts and techniques to construct $p$ should also be able to iterate those same
facts and techniques.
Thus, to advance from
$p$ to a larger ordinal which not just \emph{any} $p$-knowing
AGI could obtain, must require
the creative invention of some new facts or techniques, and
this invention requires some amount of creativity,
pattern-matching, etc. This becomes clear if the reader tries to
notate ordinals qualitatively larger than Example \ref{omegaexample};
see the more extensive examples in \cite{github}.

For analogy's sake, imagine a ladder which different AGIs
can climb, and suppose advancing up the ladder requires exercising
intelligence. One way to measure (or at least estimate) intelligence would be
to measure how high an AGI can climb said ladder.

Not all ladders are equally good. A ladder would be particularly poor if it had
a top rung which many AGIs could reach: for then it would fail to
distinguish between AGIs who could reach that top rung, even if one
AGI reaches it with ease and another with difficulty.
Even if the ladder was infinite and had no top rung, it would still be suboptimal
if there were AGIs capable of scaling the whole
ladder (i.e., of ascending however high they like, on demand)\footnote{Hibbard's
intelligence measure
\cite{hibbard2011measuring} is an infinite ladder
which is nevertheless short enough that many AGIs can
scale the whole ladder---the AGIs which
do not ``have finite intelligence'' in Hibbard's words
(see Hibbard's Proposition 3). It should be possible to
use a \emph{fast-growing hierarchy}
\cite{weiermann2002slow}
to transfinitely extend Hibbard's ladder and reduce
the set of whole-ladder-scalers. This would make
Hibbard's measurement ordinal-valued
(perhaps Hibbard intuited this; his abstract uses the
word ``ordinal''
in its everyday sense as synonym for ``natural number'').}.
A good ladder should have, for each particular AGI, a rung which that
AGI cannot reach.

Definition \ref{maindefinition} offers a good ladder.
The rungs which an AGI
manages to reach, we have argued, require core components of intelligence
to reach.
And no particular AGI can scale
the whole ladder\footnote{Thus, this ladder avoids a common problem that arises when
trying to measure machine intelligence using IQ tests, namely, that for any IQ test,
an algorithm can be designed to dominate that
test, despite being otherwise unintelligent \cite{besold2015can}.},
because no AGI can enumerate all the Intuitive Ordinal Notations: it can
be shown
that they are not computably enumerable\footnote{Namely, because if the
set of Intuitive Ordinal Notations were computably enumerable, the program $p$ which
enumerates them would itself be an Intuitive Ordinal Notation, which would
force $|p|>|p|$.}.

\subsection{Can't an AGI just print a copy of itself?}


If a truthful AGI knows its own code,
then it can certainly print a copy of itself.
But if so, then it necessarily cannot know the truthfulness of that
copy, lest it would know the truthfulness of itself.
Versions of G\"odel's incompleteness theorems adapted \cite{reinhardt1985absolute} to
mechanical knowing agents imply that a suitably idealized truthful AGI cannot know
its own code
and its own truthfulness.

\subsection{Prohibitively expensive simulation}

The reader might object that Theorem \ref{maintheorem} breaks down if $Y$ is prohibitively
expensive for $X$ to simulate. But Theorem \ref{maintheorem} and its
proof have nothing to do with simulation. In functional languages like
Haskell, functions can be manipulated, filtered,
formally composed with other functions, and so on, without needing
to be executed.
Likewise, if $X$ knows the code
of $Y$, then $X$ can manipulate and reason about that code without executing a single line
of it.



\section{Conclusion}
\label{conclusionsection}

The Intuitive Ordinal Intelligence of a truthful AGI is defined to be the supremum of the
ordinals which have Intuitive Ordinal Notations the AGI knows to be Intuitive Ordinal
Notations. We argued that this notion measures (a type of) intelligence.
We proved that if a truthful AGI single-handedly creates
a child truthful AGI, in such a way as to know the child's truthfulness and code,
then the parent must have greater Intuitive Ordinal Intelligent than the child. This
allowed us to establish a structural property for AGI populations,
resembling the Knight-Darwin
Law from biology. We speculated about implications of this biology-AGI parallel.
We hope by better understanding
how AGIs create new AGIs, we can better
understand methods of AGI-creation by humans.

\section*{Acknowledgments}

We gratefully acknowledge Jordi Bieger, Thomas Forster, Jos{\'e} Hern{\'a}ndez-Orallo,
Bill Hibbard, Mike Steel,
Albert Visser, and the reviewers for discussion and feedback.

\bibliographystyle{splncs04}
\bibliography{agikd}

\end{document}